\documentclass[a4paper,12pt]{article}
\usepackage{lmodern} 
\usepackage{amssymb,amsmath,mathrsfs}
\usepackage[dvipsnames]{xcolor}
\usepackage{hyperref}
\usepackage{enumerate}
\usepackage{array}
\usepackage{caption}
\usepackage{natbib}
\usepackage{microtype}

\usepackage[left=.8in,right=.8in,top=.8in,bottom=.8in]{geometry}

\usepackage{sectsty}
\allsectionsfont{\sffamily\mdseries\upshape}

\makeatletter
\renewenvironment{abstract}{%
    \if@twocolumn
      \section*{\abstractname}%
    \else
      \begin{center}%
        {\bfseries\sffamily\abstractname\vspace{\z@}}%
      \end{center}%
      \quotation
    \fi}
    {\if@twocolumn\else\endquotation\fi}
\makeatother

\numberwithin{equation}{section}

\hypersetup{
	colorlinks=true,
	linkcolor=MidnightBlue,
	citecolor=BrickRed,
	urlcolor=MidnightBlue
}

\newcommand{\be}{\begin{equation}}
\newcommand{\ee}{\end{equation}}

\renewcommand{\bar}{\overline}

\renewcommand{\d}{{\mathrm{d}}}

\newtheorem{prop}{Proposition}

\begin{document}

\title{The geometry-first formulation of gauge theory is not equivalent to the symmetry-first one}
\author{Henrique Gomes}

\maketitle

\begin{abstract}
This paper argues that the geometry-first and symmetry-first formulations of gauge theory are not equivalent. They differ in three respects. First, the geometry-first formulation---in which gauge groups arise as automorphism groups of structured fundamental vector bundles---admits fewer theories when its generating structures are restricted to finite tensorial data. Charged theories with additive structure group $\mathbb{R}$ have no such presentation. Second, even when a symmetry-first theory has a geometry-first presentation, its principal bundle and matter bundles do not determine which structured vector bundles generated the gauge group. Third, the natural functor from geometry-first generating objects to principal bundles with matter is faithful but neither essentially surjective nor full. Essential surjectivity fails for the diagonal quotient of the Standard Model gauge group on what I call `the classical menu', and for the additive-$\mathbb{R}$ examples on every finite tensorial menu. Fullness fails for a real oriented fibre $\mathbb{R}^{2m}$: the principal-bundle category admits the outer automorphism of $SO(2m)$ induced by an improper orthogonal map, but no structure-preserving fibre map induces it.
\end{abstract}

\section{Introduction}\label{sec:intro}
The Standard Model is formulated as a gauge theory. This paper compares two formulations of gauge theory and asks whether they determine the same theoretical structure.

The standard formulation of gauge theory begins with a structure group $G$ acting on a principal fibre bundle and constructs matter fields as sections of associated vector bundles (see \citealp{Bleecker}; \citealp{KobayashiNomizu1}; for philosophical discussion, \citealp{Weatherall2016}). A companion paper \citep{Gomes_particles}, building on \citet{Gomes_internal, Gomes_AB}, develops an alternative \emph{geometry-first} formulation. Its primitive objects are structured vector bundles equipped with compatible covariant derivatives; gauge groups arise as automorphism groups of fibres rather than being posited independently. That paper shows that much of the explanatory work usually attributed to symmetry can be done by geometry. The Higgs mechanism can be described using the extrinsic curvature of a sub-bundle singled out by the Higgs vacuum; Yukawa couplings become fibrewise contractions determined by inner-product and orientation structures on the fundamental bundles; and charge quantisation follows from the discrete algebraic structure of tensor powers of a line bundle, without appeal to the compactness of $U(1)$.

These results would still matter if the geometry-first formulation were only a notational variant of the symmetry-first one. The equivalence question determines whether the formulation changes only notation, or instead changes the primitive objects and explanatory order.

This paper addresses the equivalence question, which \citet{Gomes_particles} treated only in passing. I argue that the formulations differ in theory-space, in what can be recovered from the matter bundles, and in their categorical structure. The inclusion of theory-spaces is strict, although the direct witnesses are non-compact; for compact groups the issue is explanatory cost rather than reach. Matter bundles do not determine the geometry-first starting point: a natural recovery procedure may yield a structure group that no \emph{independent} collection of fundamental bundles generates. Finally, the canonical functor from generating objects to principal bundles with matter is faithful but neither essentially surjective nor full, under the symmetry-first picture's own standard of sameness (Section~\ref{sec:functor}).

Section~\ref{sec:setup} fixes notation and states the two formulations. Section~\ref{sec:strict} exhibits the strict inclusion. Section~\ref{sec:recovery} shows that matter bundles do not recover the geometry-first starting point. Section~\ref{sec:functor} reformulates the comparison categorically, in the framework defended by \citet{WeatherallNG2016} and \citet{BarrettHalvorson2016}, and shows that a natural functor from geometry-first to symmetry-first objects fails to be an equivalence even when the codomain is restricted to theories with compact gauge group acting faithfully on the total matter content. Section~\ref{sec:conclusion} concludes.

\section{Setup}\label{sec:setup}

This section recalls only the definitions that will be needed. Throughout, $M$ is a smooth connected manifold (spacetime), and all bundles are smooth. `PFB-POV' abbreviates `principal fibre bundle point of view' and `VB-POV' abbreviates `vector bundle point of view'. Throughout, $G$ denotes the finite-dimensional structure group---the internal gauge group---not the infinite-dimensional group of vertical bundle automorphisms; where the latter is meant, the term `gauge transformations' is used.

\paragraph{Symmetry-first (PFB-POV).} One begins with a principal $G$-bundle $(P,M,G)$ and a principal connection $\varpi$. Here $P$ is a smooth manifold carrying a free and proper right action $P\times G\to P$, $(p,g)\mapsto pg$. Matter fields are sections of associated vector bundles $\Phi_i = P\times_{\rho_i} W_i$, where $\rho_i: G\to GL(W_i)$ are representations on the vector spaces $W_i$ and $P\times_{\rho_i} W_i:=(P\times W_i)/{\sim}$ is the quotient by the equivalence
\be  \quad (p', w')\sim (p, w) \in P\times W_i \quad \text{iff}\quad (p', w')=(pg,\rho_i(g^{-1})w) \quad \text{for some}\quad g\in G,
\ee  
with $[p,w]$ denoting the class of $(p,w)$. Call these associated vector bundles \emph{matter bundles}.
The connection $\varpi$ is an equivariant $\mathfrak{g}$-valued 1-form on $P$ ($\mathfrak{g}:=\text{Lie}(G)$). It induces covariant derivatives on all matter bundles $\Phi_i$ simultaneously; the precise mechanism will not be needed here (see e.g.\ \cite{Bleecker, KobayashiNomizu1, Gomes_elements}).

\paragraph{Geometry-first (VB-POV).} One begins with a finite collection of \emph{fundamental} vector bundles $(E^a, M, V_a)$, $a=1,\ldots,N$, each equipped with additional fibre structure $\mathrm{struct}_a$ drawn from a fixed \emph{menu} of classical linear-algebraic structures: real (positive-definite, or non-degenerate of fixed signature) or Hermitian inner products $\langle\cdot,\cdot\rangle_a$, complex volume forms $\epsilon_a$, symplectic forms. In general a menu entry for a fibre $V$ is a finite list $T$ of tensors built from $V$ and $V^*$, and the associated group is the full stabiliser
\be\label{eq:menu}
\mathsf{Aut}(V,T)\;=\;\{A\in GL(V)\;:\;A\cdot t=t \ \text{ for every } t\in T\},
\ee
cut out of $GL(V)$ by polynomial equations---a real-algebraic group, with finitely many connected components. The group is the full stabiliser; the construction does not pass to identity components or impose further non-tensorial restrictions. The menu entries are fixed, smooth, non-dynamical fibre structures, analogous to fixed spacetime background structures such as a metric or foliation. The amount of such background will matter below: enlarging the menu enlarges the class of presentable gauge groups (Sections~\ref{sec:kernel} and \ref{sec:functor}). Each bundle also carries a compatible covariant derivative $\nabla^a$.\footnote{Each $\nabla^a$ determines (and is determined by) a principal connection on the bundle of admissible frames $L(E^a)$---the bundle of frames preserving the fibre structure. For instance, for a Hermitian bundle with inner product $\langle\cdot,\cdot\rangle$ and volume form $\epsilon$, an admissible frame is orthonormal and orientation-preserving, and the structure group of $L(E^a)$ is $\mathsf{Aut}(V_a, \mathrm{struct}_a)$---e.g.\ $SU(n)$ for an oriented Hermitian bundle with fibre $\mathbb{C}^n$. This correspondence relates the VB-POV to the PFB-POV; see \citet{Bleecker}, \citet{KobayashiNomizu1}, or \citet{Gomes_elements}.} The gauge group is not posited; it arises as a product of automorphism groups:
\be\label{eq:G_VB}
G = \prod_a \mathsf{Aut}(V_a, \mathrm{struct}_a).
\ee
All matter bundles are constructed from the fundamental ones by tensor products, duals, exterior powers, and direct sums, together with the sub-bundles the structures cut out canonically---images and kernels of structure-built maps, e.g.\ the traceless part of $V\otimes V^*$: the $\Phi_i$ where matter fields live must be obtained via such constructions. A field of a given particle species is a section of the bundle with fibre $W_i$, and its covariant derivative is inherited from the $\nabla^a$ by the Leibniz rule.

Note that a fundamental bundle need not contribute to any matter bundle: a gauge sector with no coupled matter---vacuum Yang--Mills---is perfectly legitimate.

\paragraph{Representational gap.} The two formulations impose different constraints on the relationship between gauge group and matter geometry. This is easiest to illustrate when there is a single fundamental fibre and a single matter bundle, $V_a=W_i=V$. In the PFB-POV, a matter bundle $\Phi = P\times_\rho V$ requires only that the representation $\rho$ preserve whatever structure $V$ carries:
\be\label{eq:subset}
\rho(G) \subseteq \mathsf{Aut}(V).
\ee
The geometry-first picture demands more. It requires
\be\label{eq:aut}
G \simeq \rho(G) \simeq \mathsf{Aut}(V).
\ee
The two isomorphisms in \eqref{eq:aut} can fail independently: $\rho(G)$ may be a proper subgroup of $\mathsf{Aut}(V)$, and $\rho$ may have a nontrivial kernel so that $G\not\simeq\rho(G)$. The gap between \eqref{eq:subset} and \eqref{eq:aut} is the representational freedom allowed by the PFB-POV and excluded by the VB-POV. In the PFB-POV, the group $G$, the representation $\rho$, and the fibre $V$ can be chosen independently, subject only to mutual consistency. In the VB-POV, the fibre structure fixes both the group and its defining action.

The two isomorphisms in \eqref{eq:aut} play different roles below. The condition $\rho(G)\simeq\mathsf{Aut}(V)$ says that the represented group is the \emph{full} stabiliser of the fibre structure. Three later arguments use this condition. Full tensor stabilisers are real-algebraic, which excludes the charged additive-$\mathbb{R}$ theories of Section~\ref{sec:strict} (Proposition~\ref{prop:alg}). No product of classical-menu stabilisers is isomorphic to the Standard Model's diagonal quotient $G/\mathbb{Z}_6$ (Sections~\ref{sec:kernel} and~\ref{sec:functor}). And a structure-preserving fibre map induces only an inner automorphism of the stabiliser, which matters for Proposition~\ref{prop:fullness}. By contrast, $G\simeq\rho(G)$ says that the action on the fibre is faithful. This holds automatically for each fundamental fibre, since a stabiliser acts faithfully on that fibre, but it need not hold for the selected matter representations. Section~\ref{sec:kernel} uses this distinction: Standard Model matter has a $\mathbb{Z}_6$ kernel, so the matter sector does not decide between $G$ and $G/\mathbb{Z}_6$.

Both formulations describe gauge-theoretic structure. The question is whether they describe the \emph{same} structure: whether one can pass between them without loss. They do not.

\section{The inclusion of theory-spaces is strict}\label{sec:strict}

Some PFB theories have no finite tensorial VB-POV presentation. Exceptional compact groups raise a different issue: they are presentable, but their presentations can use tensors that encode the group structure directly. They therefore concern explanatory cost rather than theory-space inclusion.

\paragraph{Non-integral charges.} Suppose one starts from a single fundamental bundle $(E, M, \mathbb{C}, \langle\cdot,\cdot\rangle)$---a Hermitian line bundle---and builds all charged fields tensorially: $E^{\otimes n}$ for $n\in\mathbb{Z}$. Then charge assignments necessarily form a lattice. The charge of the fundamental bundle is the generator, and all other charges are integer multiples of it. This holds regardless of whether the automorphism group of the fibre is compact ($\mathsf{Aut}(\mathbb{C},\langle\cdot,\cdot\rangle) = U(1)$) or non-compact ($\mathsf{Aut}(\mathbb{C}) = \mathbb{C}^\times$): what enforces discreteness is the tensorial construction, not the topology of the group.

In the PFB-POV with compact $U(1)$, charges are also quantised---but for a different reason: the continuous homomorphisms $U(1)\to GL(1,\mathbb{C})$ are $e^{i\theta}\mapsto e^{in\theta}$ with $n\in\mathbb{Z}$, and periodicity enforces integrality. Yet the PFB-POV does not \emph{require} compactness. If one takes the additive structure group $\mathbb{R}$, the smooth one-dimensional representations $\rho_\alpha(t) = e^{i\alpha t}$ are available for every $\alpha\in\mathbb{R}$, and one may postulate matter fields with irrational charge ratios. Among connected one-dimensional Lie groups, the alternatives are $U(1)$ and additive $\mathbb{R}$. Since the smooth characters of $U(1)$ are integral, only $\mathbb{R}$ permits arbitrary charge ratios.
Such a theory has no compact $U(1)$ replacement: tensor powers and duals of a line bundle produce an integral charge lattice, whereas charges $1$ and $\alpha\notin\mathbb{Q}$ are incommensurate.  In Proposition~\ref{prop:alg} I give the stronger statement needed below: no finite tensorial menu produces the additive $\mathbb{R}$ theory with charged matter.

Note also that prior to a choice of global group and bundle sectors the local Lagrangian density of a gauge theory sees only the Lie algebra of the gauge group: the gauge field $A_\mu$ is $\mathfrak{g}$-valued, and the covariant derivatives, field strengths, and coupling terms are built from Lie brackets and representation matrices. Since $\mathfrak{u}(1)\cong\mathbb{R}$ as Lie algebras, the Standard Model Lagrangian does not by itself distinguish $U(1)$ from $\mathbb{R}$ as the hypercharge gauge group. The specific hypercharge assignments ($1/6$, $2/3$, $-1/3$, etc.) are commensurate, so they are \emph{compatible} with the compact choice $U(1)$ but they are equally valid as representations of $\mathbb{R}$. What distinguishes the two choices is global: $U(1)$ bundles admit nontrivial topology (monopoles, topological sectors), while $\mathbb{R}$ bundles are always trivializable. The VB-POV resolves this ambiguity by construction. Moreover, the resolution is empirically commited, not merely definitional. That is, a symmetry-first theorist with gauge group $\mathbb{R}$ could absorb a newly discovered charge incommensurate with the known ones by assigning one more character; nothing local in her theory would change, for the Lagrangian depends only on the Lie algebra, which is the same. The same discovery would falsify every finite tensorial generating structure at once (Proposition~\ref{prop:alg}). The proposition thus does empirical as well as categorical work: it fixes which future charge spectra the geometry-first formulation can accommodate. The same local/global distinction appears in Section~\ref{sec:kernel}, where the Lagrangian likewise fails to distinguish $G$ from its discrete quotient $G/K$.

 It is important to note that non-compact structure groups---$SL(n,\mathbb{R})$, the Lorentz group---also arise in physical contexts, and the VB-POV admits them for the appropriate fibre structures: $\mathsf{Aut}(\mathbb{C})=\mathbb{C}^\times$, $\mathsf{Aut}(\mathbb{R}^n)=GL(n,\mathbb{R})$, and so on. Thus the additional representations allowed by the PFB-POV are not a merely formal possibility.

Of course, in the Standard Model, the internal gauge group is compact and this representational gap closes. If $G$ is compact and acts faithfully on a finite-dimensional $V$, every finite-dimensional continuous irreducible representation of $G$ is isomorphic to a subrepresentation of some $V^{\otimes r}\otimes(V^*)^{\otimes s}$, where $r$ and $s$ range over the non-negative integers. So every irreducible occurs inside some iterated tensor product of copies of $V$ and its dual.\footnote{To see this, note that products of matrix entries of a faithful representation---the functions $g\mapsto\langle w,\rho(g)v\rangle$---and their conjugates are matrix entries of the tensor powers $V^{\otimes r}\otimes(V^*)^{\otimes s}$. By the Peter--Weyl theorem, these products span a dense subspace of the continuous functions on $G$. If an irreducible representation were absent from every tensor power, its matrix entries would be orthogonal to that dense subspace, hence zero. See \citet{BrockerDieck}, Ch.~III.
Strictly, an irreducible appears as a \emph{sub}representation of a tensor power; this is why the tensorial grammar of Section~\ref{sec:setup} includes the canonically cut-out sub-bundles. The Standard Model multiplets are full tensor products of fundamentals, so the argument below does not depend on this refinement.} In the VB-POV, $G = \prod_a\mathsf{Aut}(V_a, \mathrm{struct}_a)$ acts faithfully on $\bigoplus_a V_a$ by construction, so the tensor operations of Section~\ref{sec:setup} already generate the full representation theory of $G$. The VB-POV is thus not restrictive about which matter representations can appear in this case. The obstructions of Sections~\ref{sec:recovery} and \ref{sec:functor} are of a different kind: they persist even for compact groups with faithful representations.

\paragraph{Exceptional groups.} For the classical families ($SU(n)$, $SO(n)$, $Sp(n)$), the VB-POV uses familiar fibre structures---inner products, symplectic forms, and volume forms---whose full stabilisers are the corresponding gauge groups. And though the exceptional families are also presentable, as we will see, their presentations have a different explanatory cost. In the $E_8$ case the minimal faithful representation is the adjoint, of dimension 248. Take the fibre $V=\mathfrak{e}_8$, the Lie algebra of the compact group $E_8$, with its bracket $[\cdot,\cdot]$ and the negative Killing form $-\kappa(X,Y)=-\mathrm{Tr}\bigl(Z\mapsto[X,[Y,Z]]\bigr)$. Then $\mathsf{Aut}(V,[\cdot,\cdot],\kappa)=E_8$.
This is a geometry-first presentation, but it is less explanatory in the intended sense: the primitive fibre is the gauge algebra, and the tensor that recovers the group is the Lie bracket. So, in this weaker sense, presentability is not at issue for compact groups. By Chevalley's theorem every compact group, including the exceptional groups, is the automorphism group of a faithful fibre carrying suitable tensors \citep{DeligneMilne1982}. The exceptional cases therefore concern only explanatory economy.

Call a PFB theory \emph{VB-presentable} if its structure group is an automorphism group of structured fibre input---of one or several fundamental bundles---and its matter content is generated tensorially from those bundles. Presentability can be construed as menu-relative so that unqualified, `VB-presentable' just means presentable on \emph{some} finite tensorial menu \eqref{eq:menu}; and `classically VB-presentable' means on the classical menu of Section~\ref{sec:setup}. Then the geometry-first theory-space is strictly contained in the symmetry-first one:
\be\label{eq:strict}
\{\text{VB-presentable theories}\} \subsetneq \{\text{PFB theories with matter bundles}\}.
\ee
The inclusion is strict because of the non-compact charge examples. Restricted to compact gauge groups, unqualified presentability no longer separates the formulations. However, the obstructions of Sections~\ref{sec:recovery} and \ref{sec:functor} are  not about which compact groups can be presented. They concern the provenance of a presentation and the morphisms between presentations.

\section{Matter bundles do not recover the geometry-first presentation}\label{sec:recovery}

One might concede \eqref{eq:strict} and still maintain that, for physically interesting theories like the Standard Model, the principal bundle is an idealisation that records gauge-theoretic behaviour already implicit in the matter sector. On this view, the gauge group, the connection, and the associated bundles mathematically represent how matter transforms and couples; the geometry-first construction is an intermediate representation rather than an independent input. If that is correct, one should be able to recover the geometry-first input from the matter bundles alone. The Standard Model is VB-presentable: its gauge group $SU(3)\times SU(2)\times U(1)$ arises from three fundamental bundles with fibres $\mathbb{C}^3$, $\mathbb{C}^2$, and $\mathbb{C}$. The reconstruction question is whether the converse works. Starting from a VB-POV construction---fundamental bundles $E^a$, covariant derivatives $\nabla^a$, and gauge group $G=\prod_a\mathsf{Aut}(V_a)$---then passing to the matter bundles $\Phi_i$, can one recover the same group $G$ and the same product decomposition from the matter bundles alone?

The goal is to recover a given geometry-first input from the matter sector. Section~\ref{sec:functor} formalises the forward direction as a functor $\mathcal{F}:\mathbf{VB}_{\mathrm{gen}}\to\mathbf{PFB}_{\mathrm{mat}}$; the present question is whether that passage has a canonical inverse. The natural recovery strategies fail.

The simple case of Section~\ref{sec:setup} illustrates the idea: when there is a single fundamental bundle with $V_a = W_i = V$, one can reconstruct the gauge group from the matter bundle by forming the bundle of admissible frames $L(\Phi)$. Its structure group is $\mathsf{Aut}(V, \mathrm{struct})$, and the original fundamental bundle is recovered as an associated bundle. The question is whether anything like this works when the gauge group is a product and the matter bundles are composite.

I consider three increasingly generous reconstruction attempts. Each requires information---fundamental fibres, tensorial genealogy, or factor identifications---that the matter sector does not supply. The first two attempts (\S\ref{sec:total}) fail even for VB-presentable theories; the third (\S\ref{sec:kernel}) grants the connection as well and still fails.

\subsection{A single matter bundle does not recover the gauge group}

We begin with two fundamental bundles: $(E^1, M, V_1, \langle\cdot,\cdot\rangle_1, \epsilon_1)$ with fibre $V_1 = \mathbb{C}^3$ equipped with Hermitian inner product and complex volume form, so that $\mathsf{Aut}(V_1) = SU(3)$; and $(E^2, M, V_2, \langle\cdot,\cdot\rangle_2, \epsilon_2)$ with $V_2 = \mathbb{C}^2$, so that $\mathsf{Aut}(V_2) = SU(2)$. The gauge group is $G = SU(3)\times SU(2)$. From these, we then build two matter bundles: $\Phi_1$ with fibre $W_1 = V_1$ (the defining representation of $SU(3)$, with $SU(2)$ acting trivially) and $\Phi_2$ with fibre $W_2 = V_2$ (the defining representation of $SU(2)$, with $SU(3)$ acting trivially).

Now we shall forget the fundamental bundles and try to recover $G$ from any of the single matter bundles alone. From $\Phi_1$ one reads off $\mathsf{Aut}(W_1, \langle\cdot,\cdot\rangle_1, \epsilon_1) = SU(3)$; from $\Phi_2$, $\mathsf{Aut}(W_2, \langle\cdot,\cdot\rangle_2, \epsilon_2) = SU(2)$. Each matter bundle sees only the factor of $G$ that acts on it, thus the product $SU(3)\times SU(2)$ is visible in neither alone. This illustrates the gap described in Section~\ref{sec:setup}: different factors of $G$ act trivially on different matter bundles, so no single one detects the full group. The Standard Model contains many such cases: for the right-handed electron, $SU(3)$ and $SU(2)$ both act trivially.

\subsection{The total matter bundle does not recover the group}\label{sec:total}

The previous obstruction might seem to arise only because the reconstruction considered one matter bundle at a time. But the Standard Model has many particle species, from which we can form the total matter bundle
\[
\Phi_{\mathrm{tot}} = \Phi_1 \oplus \Phi_2 \oplus \cdots \oplus \Phi_k
\]
and ask whether the full matter content recovers the gauge group. There are two natural strategies, but as we will see, both fail.

\emph{First attempt: use the full automorphism group $\mathsf{Aut}(\Phi_{\mathrm{tot}})$.} This group acts on the total fibre $W_1\oplus\cdots\oplus W_k$. It includes not only transformations that act separately on each summand, but also off-diagonal maps that mix distinct particle species---maps with components $W_i\to W_j$ for $i\neq j$; such automorphisms exist globally whatever the topology, e.g.\ $\mathrm{id}+N$ for $N$ any bundle map $\Phi_i\to\Phi_j$. These are perfectly good bundle automorphisms, but they have no counterpart in gauge theory: a gauge transformation does not turn a quark into a lepton. Thus $\mathsf{Aut}(\Phi_{\mathrm{tot}})$ is too large.

\emph{Second attempt: restrict to block-diagonal automorphisms.} Exclude the off-diagonal maps. What remains is
\[
\prod_i \mathsf{Aut}(\Phi_i),
\]
the product of the automorphism groups of the individual matter bundles. This avoids the species-mixing problem but introduces two new ones.

\emph{Invisible tensor structure.} A single composite fibre already poses a problem. The left-handed quark fibre $\mathbb{C}^3\otimes\mathbb{C}^2$ is isomorphic, as a Hermitian space, to $\mathbb{C}^6$, whose automorphism group is $U(6)$. The Standard Model group $SU(3)\times SU(2)\times U(1)$ embeds in $U(6)$, but nothing intrinsic to $\mathbb{C}^6$ singles out the tensor-product decomposition from which the product structure of $G$ derives. To recover that decomposition one needs to know which bundles are fundamental and which are composite.

\emph{Overcounting.} Even assuming that the automorphisms of each composite fibre respect its underlying tensor-product structure, a cross-species problem appears. A single gauge factor typically acts on several particle species. In the Standard Model, $SU(2)$ acts on both the left-handed quark doublet (fibre $\mathbb{C}^3\otimes\mathbb{C}^2$) and the left-handed lepton doublet (fibre $\mathbb{C}^2$). The species-indexed product $\prod_i\mathsf{Aut}(\Phi_i)$ therefore contains an independent copy of $SU(2)$ from each sector---$SU(2)_{\mathrm{quark}}$ and $SU(2)_{\mathrm{lepton}}$---while the gauge group has a single $SU(2)$ acting coherently across both. To identify these copies, one needs to know that both species descend from the same fundamental bundle $E^2$---precisely the input the VB-POV provides.

\subsection{Faithfulness is not enough: the diagonal kernel}\label{sec:kernel}

Now let us allow the reconstruction to employ all the information used in the previous attempts and some more. Suppose the tensorial genealogy of each composite fibre is known, so the tensor-product decompositions are fixed, and suppose moreover that repeated isomorphic factors across species have been correctly identified.\footnote{The second concession is already strong: it would exclude left-right symmetric models with gauge group $SU(2)_L\times SU(2)_R$, trinification models based on $SU(3)^3$, and any theory whose gauge group has isomorphic factors acting on distinct interaction sectors.} Suppose also that the full connection is available, so all parallel-transport information on the matter bundles is known. (Here I pause to note that the VB-POV derives the gauge group from fibre structure without using the covariant derivatives, so this grants the reconstruction more information than the geometry-first formulation requires.) A purely algebraic obstruction nevertheless remains.

The gap described in Section~\ref{sec:setup} has a second horn: the structure group $G$ need not act faithfully on the matter content: the total representation $\rho_{\mathrm{tot}}:G\to GL(\bigoplus_i W_i)$ may have a nontrivial kernel. When this kernel sits diagonally across the factors of $G$, the quotient $G/\ker\rho_{\mathrm{tot}}$ is a legitimate structure group that is not a product. In this case, no amount of local information, including the connection, can distinguish it from $G$ at the level of the matter sector. The natural response is to quotient by that kernel and work with the group that does act faithfully. In the Standard Model, a $\mathbb{Z}_6\subset SU(3)\times SU(2)\times U(1)$ acts trivially on every multiplet that actually occurs.\footnote{The $\mathbb{Z}_6$ kernel arises from the specific hypercharge assignments of the Standard Model particles. The convention used here is that $z\in U(1)$ acts on a field of hypercharge $Y$ as $z^{6Y}$; the Standard Model's hypercharges lie in $\frac{1}{6}\mathbb{Z}$, so all exponents are integers and the action is well defined. Elements of $SU(3)\times SU(2)\times U(1)$ of the form $(e^{2\pi i/3}\mathbf{1}_3,\, -\mathbf{1}_2,\, e^{i\pi/3})$ and their powers then act trivially on every Standard Model field.} The group that acts faithfully is the quotient
\[
G' = (SU(3)\times SU(2)\times U(1))/\mathbb{Z}_6.
\]

If one requires a faithful total action, $G'$ is the natural structure group. The issue is that $G'$ does not reproduce the \emph{independent-bundle} VB-POV presentation. Clearly, that construction produces the product $G = SU(3)\times SU(2)\times U(1)$, with each factor arising as the automorphism group of a separate fundamental fibre. In the quotient $G'$, the $\mathbb{Z}_6$ identification relates the three factors, so it cannot arise from those three independent fibres. 

But one must be careful here: it does not follow that $G'$ has no geometry-first presentation. A single rank-five Hermitian bundle $E$ with a complex volume form and a fixed background splitting $E=E_3\oplus E_2$ has automorphism group $S(U(3)\times U(2))\cong G'$, so $G'$ is produced by one structured bundle rather than three independent ones \citep{GomesWeatherall_GUT}. 
Such a fixed background splitting is a further background field in the sense of Section~\ref{sec:setup}: fixed, non-dynamical fibre structure, posited alongside the inner product and the volume form. Positing background structure to reduce an automorphism group is standard: a fixed foliation or timelike vector field on spacetime reduces general covariance in the same  way. The reconstruction problem is that the matter sector does not recall the fixed background splitting. The same multiplets are compatible with the product $G$, generated by three independent bundles $E_3,E_2,E_1$, and with the quotient $G'$, generated by one bundle with a fixed background splitting. The representations factor through $G'$ in either presentation, and no local information, including the connection, distinguishes them.\footnote{Global structure can distinguish them asymmetrically. On a topologically nontrivial $M$ the $G'$ theory admits matter-bundle sectors---twisted by the $\mathbb{Z}_6$, with fractional fluxes correlated across the factors---that cannot be produced by a product presentation, since a $G'$-bundle lifts to a $G$-bundle only if a class in $H^2(M,\mathbb{Z}_6)$ vanishes. A configuration in such a sector would refute the product provenance outright. But since every product-admissible sector is also $G'$-admissible, the matter sector can refute the independent-bundle origin but cannot confirm it. Nothing below is affected: every object in the image of $\mathcal{F}$ belongs to the untwisted sectors. See \citet{Gomes_density} on why the global form is constituted fibrewise but manifested only globally.} The reconstruction failure is therefore one related to the provenance of the matter fields: the matter fixes the action, but not the global group or the geometry that generates it.

More generally, suppose the VB-POV produces a gauge group $G = \prod_{a=1}^N G_a$, and let $K = \ker(\rho_{\mathrm{tot}}) \subset G$ be the kernel of the total action on matter. If $K$ decomposed as a product $\prod_a K_a$ with $K_a\subset G_a$, then $G/K \simeq \prod_a (G_a/K_a)$ would still be a product, and one could hope to realise each factor as the automorphism group of a modified fibre.\footnote{Even then the hope can fail within a single factor: a kernel $K_a\subset G_a$ is always normal, but it is a direct factor of $G_a$ only if $1\to K_a\to G_a\to G_a/K_a\to 1$ splits. Three cases of the total kernel should be distinguished. If $K$ is a direct factor of $G$, the vacuum-sector treatment of Section~\ref{sec:setup} applies verbatim. If $K$ is continuous but not a direct factor---for compact connected $G$ one still has $\mathfrak{g}=\mathfrak{k}\oplus\mathfrak{h}$ as a direct sum of ideals, and matter couples only to $\mathfrak{h}$---the $\mathfrak{k}$-valued connection components are locally unsourced, but $G$ need not be $K\times G/K$ globally: $U(n)\cong(U(1)\times SU(n))/\mathbb{Z}_n$. If $K$ is discrete, as the $\mathbb{Z}_6$ is, there is no vacuum sector in any dynamical sense: $G$ and $G/K$ have identical local connection forms, and $K$ records only the global form of the group. The second case yields its own essential-surjectivity witness; see footnote~\ref{fn:psu3}.} But generically the kernel sits diagonally, as the Standard Model's $\mathbb{Z}_6$ does, and then $G/K$ is not a product of automorphism groups of independent fibres: producing it geometrically requires a unified bundle carrying more background structure \citep{GomesWeatherall_GUT}. Requiring faithfulness thus does not recover the geometry-first starting point.

\section{Categorical formulation}\label{sec:functor}

Section~\ref{sec:recovery} considered reconstruction from matter bundles to geometry-first input and found no canonical reconstruction. The canonical passage goes in the other direction: from VB-POV generating objects to PFB-POV objects. Categorical equivalence has been proposed as a candidate necessary condition for, or even as constitutive of, theoretical equivalence (see \citealp{WeatherallNG2016}; \citealp{BarrettHalvorson2016}; for a review, \citealp{WeatherallCompass2019a, WeatherallCompass2019b}). The present argument need not take a stance on that debate. The categorical comparison is used here as a diagnostic.

The functor defined below is faithful but neither essentially surjective nor full. The first failure concerns objects: some PFB theories, including the faithful quotient $G'$ of the Standard Model gauge group, have no preimage among classical-menu generating objects. This is the categorical form of the diagonal-kernel obstruction in Section~\ref{sec:kernel}. The quotient $G'$ acquires a preimage once the menu admits a fixed background splitting \citep{GomesWeatherall_GUT}. Fullness fails for a different reason. The PFB category admits morphisms equivariant with respect to nontrivial automorphisms of the structure group, whereas the VB category supplies only morphisms induced by structure-preserving maps of the generating fibres. Proposition~\ref{prop:fullness} gives a real oriented $SO(2m)$ example. After the proposition, I explain the extra morphisms in terms of centralisers and normalisers.

The target category must remember matter bundles, since those are the objects from which reconstruction was supposed to proceed. The comparison therefore uses two categories: geometry-first generating objects and principal bundles with matter.

\paragraph{The source category $\mathbf{VB}_{\mathrm{gen}}$.} An object is a tuple
\[
X = \bigl((E^a, M, V_a, \mathrm{struct}_a, \nabla^a)_{a=1}^N,\; (W_i, \tau_i)_{i=1}^k\bigr),
\]
where each $(E^a, M, V_a, \mathrm{struct}_a)$ is a fundamental vector bundle carrying classical-menu structures (Section~\ref{sec:setup}) and a compatible covariant derivative, and each $W_i$ is a matter fibre equipped with a specified tensorial construction $\tau_i$ expressing $W_i$ in terms of the fundamental fibres:
\[
W_i = \tau_i(V_1,\ldots,V_N).
\]
Here $\tau_i$ is built from tensor products, duals, exterior powers, and direct sums. The construction itself, not merely the resulting vector space $W_i$, is part of the input: it determines how gauge transformations on the fundamental bundles act on $W_i$, and two constructions yielding isomorphic vector spaces may induce inequivalent representations of $G$. A morphism $f:X\to X'$---defined when $X$ and $X'$ have the same number $N$ of fundamental bundles with matching types, i.e.\ $(V_a,\mathrm{struct}_a)=(V'_a,\mathrm{struct}'_a)$ for each $a$---is a tuple $(f^a)_{a=1}^N$ of structure-preserving, connection-compatible bundle isomorphisms on the fundamental bundles, with the induced maps on all $W_i$.

The target's morphisms include a Lie group isomorphism $\eta:G\to G'$. This is the relevant target category for the symmetry-first formulation. That formulation treats the structure group abstractly: two presentations that differ only by a re-identification of that group represent the same theory, so its morphisms must include such re-identifications. Requiring $\eta=\mathrm{id}$ would compare bundles with a named group rather than gauge theories. The fullness verdict below is relative to this, the symmetry-first picture's own standard of sameness.

\paragraph{The target category $\mathbf{PFB}_{\mathrm{mat}}$.} An object is a tuple
\[
Y = \bigl(P, M, G, \varpi, (W_i,\rho_i)_{i=1}^k\bigr),
\]
consisting of a principal $G$-bundle $(P,M,G)$ with principal connection $\varpi$, together with representations $\rho_i:G\to GL(W_i)$ specifying matter bundles $\Phi_i = P\times_{\rho_i}W_i$. A morphism $(\Phi,\eta,(u_i)):Y\to Y'$ consists of a Lie group isomorphism $\eta:G\to G'$, an $\eta$-equivariant principal-bundle isomorphism $\Phi:P\to P'$ (so $\Phi(pg)=\Phi(p)\eta(g)$) with $\Phi^*\varpi' = \eta_*\circ\varpi$, and linear isomorphisms $u_i:W_i\to W_i'$ intertwining the representations:
\[
u_i\circ \rho_i(g) = \rho_i'(\eta(g))\circ u_i
\]
for all $g\in G$.

\paragraph{The functor $\mathcal{F}:\mathbf{VB}_{\mathrm{gen}}\to\mathbf{PFB}_{\mathrm{mat}}$.} From any geometry-first object $X$, form the fibre product over $M$ of the admissible-frame bundles $L(E^a)$ (as described in Section~\ref{sec:setup}) of the fundamental bundles. Then:
\[
P_X := L(E^1)\times_M\cdots\times_M L(E^N), \qquad G_X := \prod_{a=1}^N \mathsf{Aut}(V_a,\mathrm{struct}_a),
\]
with the product connection induced by the $\nabla^a$. Each tensorial construction $\tau_i$ induces a representation $\rho_i:G_X\to GL(W_i)$. Then
\[
\mathcal{F}(X) := \bigl(P_X, M, G_X, \varpi_X, (W_i,\rho_i)_{i=1}^k\bigr).
\]
On morphisms, $\mathcal{F}$ sends $(f^a)$ to the induced isomorphism of frame bundles together with the induced linear maps on all tensor constructions.

\medskip

A functor between categories is an \emph{equivalence} if and only if it is full, faithful, and essentially surjective.\footnote{A functor $F:\mathbf{C}\to\mathbf{D}$ is \emph{faithful} if it is injective on each morphism set $\mathrm{Hom}_{\mathbf{C}}(X,X')\to\mathrm{Hom}_{\mathbf{D}}(F(X),F(X'))$; \emph{full} if it is surjective on each such set; and \emph{essentially surjective} if every object of $\mathbf{D}$ is isomorphic to one in the image of $F$. These three conditions correspond, respectively, to not identifying distinct morphisms, not missing morphisms, and not missing objects. See \citet{WeatherallNG2016} for the use of categorical equivalence as a criterion of theoretical equivalence in physics, and \citet{WeatherallCompass2019a, WeatherallCompass2019b} for a review.} The following proposition shows that $\mathcal{F}$ fails two of these three conditions.

\begin{prop}\label{prop:main}
$\mathcal{F}$ is faithful but not essentially surjective, and it remains not essentially surjective when the codomain is restricted to objects in which $G$ is compact and acts faithfully on the total matter content. The compact witness is menu-relative: enlarging the structure menu to admit a fixed background splitting supplies it with a preimage. The non-compact witness survives every finite tensorial menu \eqref{eq:menu}.
\end{prop}

\emph{Faithfulness.} Suppose $\mathcal{F}(f)=\mathcal{F}(f')$ for morphisms $f,f':X\to X'$. The induced maps on the fibre product $L(E^1)\times_M\cdots\times_M L(E^N)$ then agree, so, projecting to the $a$-th factor, the induced maps $L(E^a)\to L(E'^a)$ agree for every $a$. A structure-preserving bundle isomorphism is determined by its action on admissible frames, so $f^a = f'^a$ for every $a$, hence $f=f'$.

\emph{Failure of essential surjectivity.} The functor $\mathcal{F}$ always produces a product group $G_X = \prod_a \mathsf{Aut}(V_a, \mathrm{struct}_a)$. Any object in the essential image of $\mathcal{F}$ must therefore have a structure group of this product form. This condition needs two qualifications. First, its content depends on the menu, which fixes the groups that can appear as factors---unitary, special unitary, orthogonal, symplectic, and their non-compact analogues. Second, it holds only up to abstract isomorphism, since a $\mathbf{PFB}_{\mathrm{mat}}$-isomorphism includes a Lie group isomorphism $\eta$ of structure groups. The question is therefore whether the target's structure group is isomorphic to \emph{some} product of menu groups.

Before imposing a faithfulness restriction, distinguish faithfulness of the matter action from presentability. Neither implies the other.

Non-faithful PFB theories can lie in the image of $\mathcal{F}$. An idle gauge factor---a factor acting trivially on all matter---is obtained by including a fundamental bundle from which no matter bundle is constructed, as in the vacuum sector of Section~\ref{sec:setup}. Take, for instance,
\be
X=\bigl((E^1,M,\mathbb{C}^3,\langle\cdot,\cdot\rangle_1,\epsilon_1,\nabla^1),\;(E^2,M,\mathbb{C}^2,\langle\cdot,\cdot\rangle_2,\epsilon_2,\nabla^2);\;(W,\tau)=(V_1,\mathrm{id})\bigr),
\ee
with matter built only from the first bundle. Then
\be
G_X=SU(3)\times SU(2),\qquad \rho_{\mathrm{tot}}(g_1,g_2)=g_1,\qquad \ker\rho_{\mathrm{tot}}=\{\mathbf{1}\}\times SU(2):
\ee
the image is the PFB theory with an idle $SU(2)$ factor. More generally, a principal $G_1\times G_2$-bundle is a fibre product of a $G_1$-bundle and a $G_2$-bundle, and every principal $SU(k)$-bundle is the admissible-frame bundle of a Hermitian bundle with volume form, including the connection. Partial kernels can also arise from the tensorial constructions themselves. With a single fundamental bundle $(\mathbb{C}^2,\langle\cdot,\cdot\rangle,\epsilon)$ and matter only the traceless conjugation construction,
\be
W=(V\otimes V^*)_0,\qquad \rho(g)A=gAg^{-1},\qquad \ker\rho=\{\pm\mathbf{1}\}=Z(SU(2)),
\ee
the image is the $SU(2)$ theory with adjoint matter, non-faithful by its centre. In general the total kernel,
\be
\ker\rho_{\mathrm{tot}}=\bigcap_i\ker\rho_i,
\ee
is fixed by which constructions the object includes, and for compact menu-product groups the Peter--Weyl argument of Section~\ref{sec:strict} makes every representation available tensorially. The Standard Model \emph{product} theory itself, with its $\mathbb{Z}_6$ kernel, is the image of the standard three-bundle object. Thus non-faithfulness alone does not obstruct presentability. Requirement \eqref{eq:aut} constrains each group factor through its fundamental fibre, not through the selected matter representations: the idle factor above still satisfies $G_2=\mathsf{Aut}(V_2,\mathrm{struct}_2)$ and has its own Yang--Mills sector. The two witnesses below concern different failures: groups that no permitted fibre presents, and representations that no permitted tensor construction produces.

Conversely, faithfulness does not secure presentability. Restricting the codomain to compact objects with faithful total matter action strengthens the result: essential surjectivity still fails after idle factors and matter kernels have been excluded. Let $\mathbf{PFB}^*_{\mathrm{mat}}$ denote the full subcategory of $\mathbf{PFB}_{\mathrm{mat}}$ consisting of objects where $G$ is compact and the total representation $\rho_{\mathrm{tot}}:G\to GL(\bigoplus_i W_i)$ is faithful. Compactness removes the representational obstruction of Section~\ref{sec:strict}, since tensor powers of a faithful representation of a compact group generate every irreducible representation. The remaining failure is therefore not caused by the tensorial grammar omitting some representation.

Even on $\mathbf{PFB}^*_{\mathrm{mat}}$, $\mathcal{F}$ is not essentially surjective. The Standard Model supplies a witness. The VB-POV produces $G = SU(3)\times SU(2)\times U(1)$, but a $\mathbb{Z}_6\subset G$ acts trivially on all Standard Model multiplets. The quotient $G' = G/\mathbb{Z}_6$ is compact and acts faithfully on the total matter content, so it defines an object of $\mathbf{PFB}^*_{\mathrm{mat}}$. But $G'$ is not isomorphic to any product of menu groups with Lie algebra $\mathfrak{su}(3)\oplus\mathfrak{su}(2)\oplus\mathfrak{u}(1)$.\footnote{Since $G'$ is connected, factors with disconnected automorphism groups ($O(n)$, and the finite $O(1)=\mathbb{Z}_2$) are excluded, and the candidates are $SU(3)\times SU(2)\times U(1)$, $SU(3)\times U(2)$, $U(3)\times SU(2)$---the global forms $\Gamma = 1, \mathbb{Z}_2, \mathbb{Z}_3$ \citep{Tong2017, GomesWeatherall_GUT}---and variants with $SO(3)$ as the weak factor. Centres separate the first three from $G'$: the centre of $G'\cong S(U(3)\times U(2))$ consists of the pairs $(\alpha\mathbf{1}_3,\beta\mathbf{1}_2)$ with $\alpha^3\beta^2=1$, a \emph{connected} group, isomorphic to $U(1)$ because $\gcd(3,2)=1$; the candidates' centres are $U(1)\times\mathbb{Z}_6$, $U(1)\times\mathbb{Z}_3$, and $U(1)\times\mathbb{Z}_2$ respectively, each disconnected by its finite factor. Fundamental groups exclude the $SO(3)$ variants: $\pi_1(G')\cong\mathbb{Z}$, generated by any lift of the hypercharge circle through the $\mathbb{Z}_6$ identification, contains no element of finite order, while an $SO(3)$ factor contributes one of order two, since $\pi_1(SO(3))=\mathbb{Z}_2$.
} Hence no object in the essential image is isomorphic to the $G'$ theory.\footnote{\label{fn:psu3}The diagonal $\mathbb{Z}_6$ is not the only blocking mechanism. A second, independent witness arises inside a single factor, from a continuous kernel. With a single fundamental bundle $(\mathbb{C}^3,\langle\cdot,\cdot\rangle)$---no volume form, so the recovered group is $U(3)$---and matter only in the traceless conjugation construction, the kernel is the centre $U(1)\subset U(3)$: continuous, invisible to the matter covariant derivative, but not a direct factor, since $U(3)\cong(U(1)\times SU(3))/\mathbb{Z}_3$ does not split. The faithful quotient $PU(3)\cong PSU(3)$ is compact, so it defines an object of $\mathbf{PFB}^*_{\mathrm{mat}}$ with no classical-menu preimage: it is connected, centreless, and simple, hence not a nontrivial product of menu groups; the only compact classical simple Lie algebra of dimension eight is $\mathfrak{su}(3)$; and the classical-menu group with that algebra is $SU(3)$, whose centre $\mathbb{Z}_3$ is nontrivial. The blocking mechanism is a non-split extension inside a single factor, with no diagonal embedding anywhere. (Contrast the $SU(2)$ adjoint example above, whose faithful quotient $SO(3)$ is classically presentable as $\mathsf{Aut}(\mathbb{R}^3,\langle\cdot,\cdot\rangle,\mathrm{vol})$.) The exclusion is again menu-relative: $PSU(3)$ is presentable off-menu with fibre $\mathfrak{su}(3)$ carrying the bracket and the symmetric cubic $d(X,Y,Z)=\mathrm{Re}\,\mathrm{Tr}(XYZ+XZY)$. The bracket alone yields $\mathsf{Aut}(\mathfrak{su}(3))\cong PSU(3)\rtimes\mathbb{Z}_2$; the outer automorphism flips the sign of $d$, so fixing the cubic cuts the stabiliser to $PSU(3)$.}

This is a failure relative to the classical menu, not to geometry-first presentations as such. The quotient $G'\cong S(U(3)\times U(2))$ is the automorphism group of a single rank-five Hermitian bundle with a complex volume form and a fixed background $3\oplus2$ splitting \citep{GomesWeatherall_GUT}. Enlarging the menu to admit the splitting---one structured fibre, carrying more background structure, in place of three independent ones---restores a preimage.\footnote{The classical menu does offer a dynamical construction of $G'$: begin with an $SU(5)$ fibre and let an adjoint Higgs supply the splitting rather than treating it as fixed background structure. This construction explains the splitting that the present $G'$ theory posits; the hypercharge locking follows from the split bundle in either case \citep{Gomes_density}, and \citet{GomesWeatherall_GUT} compares the two. It does not, however, provide a preimage of the $G'$ theory. The functor assigns the unbroken object the structure group $SU(5)$, not $G'$, and the Higgsed theory contains heavy off-diagonal gauge modes absent from the $G'$ theory. It also restricts the light chiral matter to complete packages obtained by decomposing tensors of the rank-five fibre, whereas a fixed background splitting admits any locked collection.} Thus the matter sector underdetermines the gauge group's provenance: the symmetry-first formulation does not record whether the group is the product, generated by independent bundles, or the quotient, generated by a single bundle with a fixed background splitting.

The non-compact witness is not menu-relative. Section~\ref{sec:strict} showed that matter generated tensorially from a line has charges $q(E^{\otimes n})=nq_0$, and hence a lattice spectrum. To turn this difference into a one-parameter witness, choose a theory whose charge spectrum is not a lattice. The connected one-dimensional Lie groups are $U(1)$ and additive $\mathbb{R}$; since the characters of $U(1)$ are integral, only additive $\mathbb{R}$ permits arbitrary real charge ratios. Its smooth characters $\rho_q(t)=e^{iqt}$ allow every real $q$.

Two claims must then be kept separate. Proposition~\ref{prop:alg} shows that no finite tensorial menu presents additive $\mathbb{R}$ with nontrivial charged matter. The paragraph following the proposition shows that, when two charges are incommensurate, no compact $U(1)$ theory has the same charge spectrum. The missing preimage therefore cannot be dismissed merely as a difference in global form.
\begin{prop}[Additive $\mathbb{R}$ excluded]\label{prop:alg}
Let $Y$ be a PFB theory whose structure group is the additive group $\mathbb{R}$ and whose matter content includes a one-dimensional representation $\rho_q(t)=e^{iqt}$ with $q\neq0$. Then $Y$ has no preimage under $\mathcal{F}$, for any finite tensorial menu \eqref{eq:menu}.
\end{prop}
\emph{Proof.} A preimage $X$ would have recovered group $G_X=\mathsf{Aut}(V,T)$ isomorphic to $\mathbb{R}$ as a Lie group, since an isomorphism in $\mathbf{PFB}_{\mathrm{mat}}$ includes a Lie group isomorphism $\eta$ of structure groups. Three facts, detailed in Appendix~\ref{app:algebraic-charges}, rule this out. First, $G_X$ is \emph{real-algebraic}: for each tensor $t\in T$ the condition $A\cdot t=t$ is polynomial in the matrix entries of $A$, and \eqref{eq:menu} takes the full solution set. Second, every matter representation built by the tensorial grammar of Section~\ref{sec:setup} is \emph{algebraic}: its matrix entries are polynomials in the entries of $g$ and $g^{-1}$. Third, the only one-dimensional real-algebraic group isomorphic to $\mathbb{R}$ as a Lie group is the additive algebraic group, and the additive group has no nontrivial algebraic character: an algebraic character is a polynomial $\chi$ satisfying
\[
\chi(s+t)=\chi(s)\,\chi(t),
\]
so $\chi(2t)=\chi(t)^{2}$, and comparing degrees gives $d=2d$, hence $d=0$ and $\chi\equiv1$. The matter representation $\rho_q$ with $q\neq0$ is a nontrivial character---smooth, but not algebraic---so no tensorial presentation carries it. \hfill$\square$

Proposition~\ref{prop:alg} strengthens the lattice argument of Section~\ref{sec:strict} by removing any restriction to the classical menu. Every one-dimensional recovered group has a tensorial character lattice of rank at most one: rank one for $U(1)$ and $\mathbb{R}^{\times}$, whose algebraic characters are the integer powers $z\mapsto z^{n}$, and rank zero for the additive group. The obstruction is therefore stronger than the failure to place two incommensurate charges in one lattice.

\paragraph{No compact replacement.} The absence of a preimage would be less significant if every excluded theory had a compact replacement: a VB-presentable theory with the same Lie algebra and local content, differing only in global form. For commensurate charges such a replacement exists---the $U(1)$ theory of Section~\ref{sec:strict}---so there the missing preimage can be read as the VB-POV correcting the global form of the gauge group. For incommensurate charges there is no compact replacement, by the lattice arithmetic of tensorial charge described above. Again, this is easy to see, as follows. Fix a fundamental Hermitian line $E$ and let $q_0$ be its charge: the character by which the recovered group $U(1)$ acts on the fibre. The grammar generates the charges
\[
q(E^{\otimes n})=n\,q_0,\qquad q(E^{*})=-\,q_0,\qquad n\in\mathbb{Z}:
\]
tensor products add charges and duals reverse them, so every available charge lies in the lattice $q_0\mathbb{Z}$. If $1$ and $\alpha$ both lie in $q_0\mathbb{Z}$, then $1=nq_0$ and $\alpha=mq_0$ for integers $n,m$, so $\alpha=m/n\in\mathbb{Q}$. Incommensurate charges therefore cannot fit in a single lattice. 

A failure of essential surjectivity alone could be removed by restricting the codomain to the essential image. Proposition~\ref{prop:fullness} shows that it would not suffice: fullness fails at an object inside the essential image.

\begin{prop}[Failure of fullness]\label{prop:fullness}
The functor $\mathcal{F}:\mathbf{VB}_{\mathrm{gen}}\to\mathbf{PFB}_{\mathrm{mat}}$ is not full.
\end{prop}

\emph{Proof.} Recall the morphisms on the two sides. In $\mathbf{VB}_{\mathrm{gen}}$, a morphism is a tuple $(f^a)$ of structure-preserving, connection-compatible bundle isomorphisms on the fundamental bundles. In $\mathbf{PFB}_{\mathrm{mat}}$, a morphism is a triple $(\Phi,\eta,(u_i))$, where $\eta$ is a Lie group isomorphism of the structure groups, $\Phi$ is an $\eta$-equivariant principal-bundle isomorphism, and the $u_i$ intertwine the matter representations.

First, every morphism in the image of $\mathcal{F}$ has trivial $\eta$-component. Let $f=(f^a)_{a=1}^N:X\to X'$ be a morphism in $\mathbf{VB}_{\mathrm{gen}}$. Fix $a$, and let $L(E^a)$ be the admissible-frame bundle of $E^a$, with structure group $G_a=\mathsf{Aut}(V_a,\mathrm{struct}_a)$. A point of $L(E^a)$ over $x\in M$ is an admissible frame $e:V_a\to E^a_x$. The induced map on frame bundles is $\Phi_{f^a}(e)=f^a_x\circ e$. For any $g\in G_a$,
\[
\Phi_{f^a}(e\cdot g)=\Phi_{f^a}(e\circ g)=f^a_x\circ e\circ g=(f^a_x\circ e)\circ g=\Phi_{f^a}(e)\cdot g.
\]
Thus $\Phi_{f^a}$ is $G_a$-equivariant with respect to the identity automorphism of $G_a$. Taking the product over $a$, every morphism $\mathcal{F}(f)$ has $\eta=\mathrm{id}_{G_X}$. It remains to exhibit a PFB morphism between objects in the image of $\mathcal{F}$ whose $\eta$-component cannot arise from a VB morphism.

Inner automorphisms do not give the relevant fullness failure. If $\eta_h(g)=hgh^{-1}$ and $(\Phi,\eta_h,u)$ is a PFB morphism, then setting $\Phi'(p)=\Phi(p)h$ and $u'=\rho(h^{-1})\circ u$ gives an equivalent description with $\eta=\mathrm{id}$ on the associated matter bundles.\footnote{Equivariance with $\eta=\mathrm{id}$ follows from $\Phi'(pg)=\Phi(pg)h=\Phi(p)\eta_h(g)h=\Phi(p)hgh^{-1}h=\Phi'(p)g$. On matter bundles, $[\Phi'(p),u'w]=[\Phi(p)h,\rho(h^{-1})uw]=[\Phi(p),uw]$ by the equivalence relation defining associated bundles.} The witness must therefore use an outer automorphism.

Take a geometry-first object
\[
X_{\mathbb{R}}=(E\to M,\; V=\mathbb{R}^{2m},\; g,\;\mathrm{vol},\;\nabla),\qquad m\geq 2,
\]
where $g$ is a positive-definite inner product and $\mathrm{vol}$ is a volume form. The recovered group is
\[
G_X=\mathsf{Aut}(\mathbb{R}^{2m},g,\mathrm{vol})=SO(2m).
\]
Include the defining matter fibre $W=V$, with $\rho$ the defining representation.

Choose a reflection $R\in O(2m)$, so $\det R=-1$, and define
\[
\eta(g)=RgR^{-1}.
\]
Conjugation by $R$ preserves orthogonality and determinant, so $\eta$ is an automorphism of $SO(2m)$. It is outer. If $\eta$ were inner, then for some $h\in SO(2m)$,
\[
RgR^{-1}=hgh^{-1}\qquad\text{for all }g\in SO(2m).
\]
Then $h^{-1}R$ would commute with every element of the defining representation of $SO(2m)$. For $2m\geq4$, the commutant of this representation consists of the scalar maps, so $h^{-1}R=\lambda\mathbf{1}$. Hence $R=\lambda h$. Taking determinants gives
\[
-1=\det R=\lambda^{2m}\det h=\lambda^{2m},
\]
which has no real solution. Thus $\eta$ is outer.\footnote{For $m=1$, $SO(2)$ is abelian, so conjugation by a reflection is the inversion automorphism $g\mapsto g^{-1}$, which is outer because all inner automorphisms of an abelian group are trivial; the commutant argument is not needed. Under the morphisms of this section---linear, structure-preserving---$m=1$ would therefore also serve. The restriction to $m\geq2$ rules out the semilinear reply discussed in Appendix~\ref{app:complex}. Since $SO(2)\cong U(1)$, the group also has the classical presentation given by a Hermitian line, and complex conjugation implements inversion at that presentation. For $m\geq2$ no classical-menu stabiliser on a complex fibre is isomorphic to $SO(2m)$, so no analogous classical presentation exists.
}

On the trivial principal bundle $P=M\times SO(2m)$ with flat connection, define
\[
\Phi(x,g)=(x,RgR^{-1}).
\]
This map is $\eta$-equivariant:
\[
\Phi((x,g)h)=(x,RghR^{-1})=(x,RgR^{-1}RhR^{-1})=\Phi(x,g)\eta(h).
\]
It preserves the flat connection because $\eta$ is a Lie group automorphism.\footnote{For the Maurer--Cartan form, $\Phi^*(g^{-1}\d g)=\eta_*(g^{-1}\d g)$.} On the defining matter fibre set
\[
u=R.
\]
The intertwining condition is
\[
u\rho(g)=\rho(\eta(g))u,
\]
which reads
\[
Rg=(RgR^{-1})R.
\]
Therefore $(\Phi,\eta,u):\mathcal{F}(X_{\mathbb{R}})\to\mathcal{F}(X_{\mathbb{R}})$ is a morphism in $\mathbf{PFB}_{\mathrm{mat}}$.

No VB morphism induces this one. Every morphism in the image of $\mathcal{F}$ has $\eta=\mathrm{id}$, while here $\eta=\mathrm{Ad}_R$, where $\mathrm{Ad}_R(g):=RgR^{-1}$. Nor can the inner shift described above convert it into an image morphism. A shift $\Phi\mapsto\Phi\cdot h$, $u\mapsto\rho(h^{-1})\circ u$ with $h\in SO(2m)$ replaces $\eta$ by $\mathrm{Ad}_{h^{-1}}\circ\eta$, since
\[
(\Phi\cdot h)(pg)=\Phi(p)\,\eta(g)\,h=(\Phi\cdot h)(p)\cdot\bigl(h^{-1}\eta(g)h\bigr).
\]
The shift changes $\eta$ only within its outer class, and $[\mathrm{Ad}_R]\neq[\mathrm{id}]$ because $\eta$ is outer. Hence no re-presentation of $(\Phi,\eta,u)$ lies in the image of
\[
\mathrm{Hom}_{\mathbf{VB}_{\mathrm{gen}}}(X_{\mathbb{R}},X_{\mathbb{R}})\;\longrightarrow\;\mathrm{Hom}_{\mathbf{PFB}_{\mathrm{mat}}}\bigl(\mathcal{F}(X_{\mathbb{R}}),\mathcal{F}(X_{\mathbb{R}})\bigr),
\]
and $\mathcal{F}$ is not full. \hfill$\square$

Two natural attempts to supply the missing morphism at $X_{\mathbb{R}}$ fail. The direct candidate is the fibre map $v=R$. It preserves $g$ but not the volume form:
\[
R^{*}\mathrm{vol}=(\det R)\,\mathrm{vol}=-\mathrm{vol}.
\]
So $R$ is not an automorphism of $X_{\mathbb{R}}$; it is a morphism from $X_{\mathbb{R}}$ to the object carrying the opposite orientation, and fullness concerns the Hom-set at $X_{\mathbb{R}}$. The second attempt enlarges the morphisms to semilinear maps, as one can do over complex fibres (Appendix~\ref{app:complex}). A $\sigma$-semilinear map satisfies $v(\lambda x)=\sigma(\lambda)v(x)$ for a field automorphism $\sigma$ of the scalars. Over $\mathbb{R}$ this adds nothing. Every field automorphism of $\mathbb{R}$ fixes $\mathbb{Q}$, preserves squares and hence order, and therefore fixes every real number. Every candidate fibre map is linear. If it preserves $(g,\mathrm{vol})$, it belongs to $SO(2m)$ and induces an inner automorphism. Thus the fibre structure that determines the group also excludes the improper automorphism from the source Hom-set. One could instead build a different source category that admits improper re-identifications as morphisms; that changes the source-side notion of structure preservation rather than making $\mathcal{F}$ full.

One can also try to weaken the object rather than enlarge the morphisms: replace the volume form by the unordered pair
\[
[\mathrm{vol}]=\{\mathrm{vol},-\mathrm{vol}\}.
\]
Then $R$ would preserve the coarsened orientation datum. But the coarsened object does not recover the original group. Any linear map preserving $g$ already has determinant $\pm1$: in matrix notation $v^TGv=G$, and taking determinants gives $(\det v)^2=1$. Thus every $g$-preserving map sends $\mathrm{vol}$ to $\pm\mathrm{vol}$. Preserving $(g,[\mathrm{vol}])$ is therefore the same as preserving $g$ alone, and the recovered group becomes
\[
\mathsf{Aut}(V,g)=O(2m),
\]
not $SO(2m)$. The full stabiliser of $(g,[\mathrm{vol}])$ is therefore the same as the full stabiliser of $g$; the sign-orbit adds no further stabiliser information.

\paragraph{Centralisers and normalisers.}
At the $SO(2m)$ witness, centralisers and normalisers classify the additional target morphisms. For a represented group $G\subseteq GL(V)$, define
\be\label{eq:cent-norm}
\begin{aligned}
C_{GL(V)}(G)&=\{u\in GL(V)\;:\;ug=gu\ \text{for all}\ g\in G\},\\
N_{GL(V)}(G)&=\{u\in GL(V)\;:\;uGu^{-1}=G\}.
\end{aligned}
\ee
An element of the centraliser intertwines the representation with $\eta=\mathrm{id}$. An element of the normaliser intertwines it after the group has been re-identified by $\eta=\mathrm{Ad}_u$. For the trivial flat object used here, every target automorphism arises from a normaliser element.

Let $(\Phi,\eta,u)$ be an automorphism of $\mathcal{F}(X_{\mathbb{R}})$ in $\mathbf{PFB}_{\mathrm{mat}}$. Because the matter representation is faithful and defining, the intertwining condition gives $\eta(g)=ugu^{-1}$. Hence $u$ belongs to $N:=N_{GL(V)}(SO(2m))$ and $\eta=\mathrm{Ad}_u$. Conversely, every $u\in N$ defines such an automorphism on the trivial bundle by $\Phi(x,g)=(x,ugu^{-1})$. Composing with a geometry-side automorphism $h\in SO(2m)$ replaces $u$ by $h^{-1}u$. Since $SO(2m)$ is normal in $N$, the classes not removed in this way form $N/SO(2m)$. Morphisms in the image of $\mathcal{F}$ correspond to the identity class.\footnote{This description is restricted to the present example: the bundle is trivial, the connection is flat, and the matter list contains one faithful defining representation. For a single faithful representation there is an exact sequence $1\to C_{GL(V)}(G)/Z(G)\to N_{GL(V)}(G)/G\to \mathrm{Out}(G)_{[V]}\to 1$, where $Z(G)=C_{GL(V)}(G)\cap G$ and $\mathrm{Out}(G)_{[V]}$ consists of the outer classes realised by conjugation in $GL(V)$. With several representations, the target may also contain independent elements of their commutants.}

The commutant of the defining representation of $SO(2m)$ consists of scalar maps. If $u\in N$, then the form
\[
(v,w)\longmapsto g(u^{-1}v,u^{-1}w)
\]
is $SO(2m)$-invariant and positive-definite. The defining representation admits a unique invariant positive-definite form up to scale, so this form equals $\lambda g$ for some $\lambda>0$. Thus $u$ is a positive multiple of an orthogonal map, and
\be\label{eq:normaliser}
N_{GL(2m,\mathbb{R})}\bigl(SO(2m)\bigr)=\mathbb{R}^{>0}\cdot O(2m),
\qquad
N/SO(2m)\;\cong\;\mathbb{R}^{>0}\times\mathbb{Z}_{2}.
\ee
The positive-scalar factor rescales the inner product and has $\eta=\mathrm{id}$. The $\mathbb{Z}_{2}$ factor reverses orientation and induces the outer automorphism $\mathrm{Ad}_R$.

The scalar factor gives a shorter proof of non-fullness. For $\lambda\neq\pm1$, $(\mathrm{id},\mathrm{id},\lambda\mathbf{1})$ is an automorphism of $\mathcal{F}(X_{\mathbb{R}})$ with no geometry-first preimage. I retain the reflection as the main witness because this scalar argument depends on allowing arbitrary linear matter intertwiners. If the target category also records the induced inner products and requires each $u_i$ to preserve them, the scalar maps disappear, whereas the reflection remains. This removes the positive-scalar factor but not the orientation-reversing class, so Proposition~\ref{prop:fullness} survives the strengthened morphism condition. This is the categorical analogue of the commutant freedom in the Yukawa analysis of \citet{Gomes_particles}. There symmetry fixes the Yukawa map only up to insertions from the commutant, while the geometry-first fibre structures select a particular contraction; here the same freedom enlarges a Hom-set.

For the unoriented object $(\mathbb{R}^{2m},g)$, the recovered group is $O(2m)$ and the same calculation gives $N/O(2m)\cong\mathbb{R}^{>0}$. The reflection now belongs to the group, so only the scale ambiguity remains. Adding $\mathrm{vol}$ changes the stabiliser from $O(2m)$ to $SO(2m)$ and places $R$ in the normaliser but outside the group. The status of the improper map therefore depends on the declared fibre structure. The same positive-scalar ambiguity appears when one reconstructs a spacetime metric from local Lorentz transition groups: the normaliser of the orthogonal group determines the metric only up to a conformal factor, which a volume element fixes \citep{GMPR_charts}. For $O(2m)$ the quotient contains only the scalar factor; for $SO(2m)$ it also contains the orientation-reversing class.

Proposition~\ref{prop:fullness} concerns the endomorphisms of the real oriented object $X_{\mathbb{R}}$; a different presentation may implement the same improper automorphism. Appendix~\ref{app:complex} gives a complexified presentation with a suitable semilinear morphism. This does not restore fullness: the PFB images of the real and complex presentations are isomorphic although the source objects are not, so the functor is not full on the Hom-set between them. The appendix also shows that, for the real and complex scalar fields used here, semilinear enlargements realise only outer classes of order at most two at any fixed presentation. Table~\ref{tab:presentation-morphisms} lists the relevant PFB isomorphisms.

Combining Propositions~\ref{prop:main} and~\ref{prop:fullness}, $\mathcal{F}$ is faithful but neither essentially surjective nor full, and hence is not an equivalence. Enlarging the source menu can supply preimages for particular objects, and enlarging the source morphisms can supply particular missing arrows. But as long as the target retains charged additive-$\mathbb{R}$ theories, no finite tensorial enlargement of the source can make $\mathcal{F}$ essentially surjective, by Proposition~\ref{prop:alg}. An equivalence would therefore also require restricting the target category.

The two negative results are independent. Essential surjectivity concerns objects: some PFB theories have no preimage on the specified geometry-first menu. The diagonal quotient $G'$ gives a compact menu-relative case, while the non-compact incommensurate-charge theories of Section~\ref{sec:strict} give cases excluded by every finite tensorial menu. Fullness concerns morphisms: even when a PFB object lies in the image of $\mathcal{F}$, that image object may have PFB automorphisms with no geometry-first preimage.

It is sometimes argued that failure of essential surjectivity is philosophically uninteresting: if a functor is full and faithful but not essentially surjective, one can restrict the codomain to the essential image and recover equivalence (\citealp{WeatherallNG2016}; for discussion, \citealp{Shi2025}). That response does not apply here, because $\mathcal{F}$ is not full. Nor would such a restriction be philosophically neutral: it would define the symmetry-first category by the VB-POV's admissibility conditions. There is also a physical point beside the methodological one. The theories the restriction discards encode a possible course of experience: a symmetry-first theorist who retains the charged additive-$\mathbb{R}$ theories can describe a world in which the next measured charge is incommensurate with the known ones without fundamentally altering her Lagrangian, whereas the restricted category, like the geometry-first category it would imitate, cannot. Restricting the codomain therefore changes the framework's modal empirical content---which future spectra it can represent---and not merely its definitional book-keeping.

\section{Conclusion}\label{sec:conclusion}

The geometry-first and symmetry-first formulations of gauge theory should not be called equivalent without a good deal of provisos. Firstly, they differ in theory-space: the VB-POV applies only to theories whose matter sector is generated tensorially from appropriate fundamental bundles, which excludes PFB-POV theories with non-compact gauge groups. Second ,they differ in recoverable structure: a principal bundle together with its matter bundles does not record which bundles are fundamental, nor how the matter content is generated from them. Lastly, they also differ in explanatory order: the geometry-first formulation derives the gauge group from fibre structure, while the symmetry-first formulation posits the structure group independently.

The categorical comparison gives the formal counterpart of these differences. The functor $\mathcal{F}:\mathbf{VB}_{\mathrm{gen}}\to\mathbf{PFB}_{\mathrm{mat}}$ is faithful, but it is neither essentially surjective nor full. Essential surjectivity fails in two ways. The compact Standard Model quotient $G/\mathbb{Z}_6$ is not produced by separate fundamental bundles on the classical menu, though it is geometry-first presentable by a single bundle carrying a fixed background splitting \citep{GomesWeatherall_GUT}. The non-compact incommensurate-charge theories have no preimage on any finite tensorial menu. And fullness fails already on the classical menu: the real oriented object with fibre $\mathbb{R}^{2m}$ recovers $SO(2m)$, but the PFB side admits an automorphism of the image object induced by an improper orthogonal map. The VB side has no corresponding morphism, because the reflection reverses the volume form, while every map preserving both $g$ and $\mathrm{vol}$ lies in $SO(2m)$ and induces only an inner automorphism. Weakening orientation to a sign-orbit changes the recovered group to $O(2m)$. 

The geometry-first reformulation defended in \citet{Gomes_particles} is therefore not a notational variant of the principal-bundle formalism. It uses stricter primitive objects and a different explanatory order. For compact groups, the main question is often one of geometric provenance and economy rather than mere presentability. The PFB-POV takes the gauge group as a single Lie group $G$; it does not require $G$ to be a product, does not specify which factor comes from which fundamental fibre, and does not encode the tensorial genealogy of the matter representations. The VB-POV records this interaction structure directly: each fundamental bundle defines an interaction sector, particle species interact through a force when their matter fibres share the corresponding fundamental factor, and the gauge group is read off from the fibre structures. The non-equivalence established here is the claim that this interaction structure is not recoverable from the symmetry-first objects alone.

I should note that physicists already impose some of these restrictions in practice. For instance, a gauge factor that acts trivially on all matter is formally allowed in both formulations, but it is usually omitted because it has no empirical role in the matter sector. \citet{Skinner_QFT} is unusually explicit about another restriction. After constructing principal bundles from frame bundles under the assumption $G\simeq\mathsf{Aut}(V)$, looking at a single matter bundle he describes the two viewpoints as equivalent for the matrix Lie groups common in physics, but treats the principal-bundle formulation as more fundamental for exceptional groups. This qualification identifies the assumptions at issue here: classical groups acting faithfully on fundamental vector spaces, with matter generated tensorially from them. If these assumptions define the intended class of models, they should be stated rather than treated as automatic. The geometry-first formulation states them.

\paragraph{Acknowledgements.} I thank Jim Weatherall for extensive correspondence on this material, and especially for pressing the conjugation/semilinear reading that motivates Appendix~\ref{app:complex}.

\clearpage
\appendix

\section{Algebraic charge obstruction}\label{app:algebraic-charges}

This appendix completes the proof of Proposition~\ref{prop:alg}: no finite tensorial menu, classical or otherwise, recovers the additive group $\mathbb{R}$ with charged matter. (The lattice argument of Section~\ref{sec:functor} establishes the separate claim that incommensurate charges also have no compact replacement; I do not repeat it here.)

A subgroup of $GL(V)$ is \emph{real-algebraic} if it is the full solution set of a system of polynomial equations in the real matrix entries. For a complex fibre one uses the real and imaginary parts of the entries. The stabiliser \eqref{eq:menu} is real-algebraic for every finite tensorial menu, since the condition $A\cdot t=t$ is, for each tensor $t$, a finite system of polynomial equations in the entries of $A$. Unitarity, for example, reads $\bar A^T A=\mathbf{1}$, polynomially in the real coordinates.

A representation is \emph{algebraic} when its matrix entries are polynomial functions of the entries of $g$ and $g^{-1}$. The tensorial grammar of Section~\ref{sec:setup} produces only algebraic representations: tensor products, duals, exterior powers, direct sums, and structure-built images or kernels are all described by polynomial formulas in $g$ and $g^{-1}$.

A one-dimensional representation $\chi:G\to\mathbb{C}^{\times}$ is a character. For abelian one-dimensional matter, the character encodes the charge: a field of charge $q$ for additive $\mathbb{R}$ transforms by $t\mapsto e^{iqt}$. Tensor products multiply characters, which adds charges; duals invert characters, which changes the sign of the charge. Hence a finite tensorial presentation produces an integral character lattice.

Now suppose a finite tensorial presentation recovered a group isomorphic to additive $\mathbb{R}$. Since the recovered group is a full tensor stabiliser, it is real-algebraic; since it is isomorphic to $\mathbb{R}$, it is also connected and one-dimensional. Up to real-algebraic isomorphism, a connected one-dimensional real linear algebraic group is the additive group, the split multiplicative group, or the compact torus; their real points are $(\mathbb{R},+)$, $\mathbb{R}^{\times}$, and $U(1)$, respectively \citep{Borel1991}. Only the additive group is isomorphic to $\mathbb{R}$ as a Lie group: $\mathbb{R}^{\times}$ is disconnected and $U(1)$ is compact. The identity component $\mathbb{R}_{>0}\subset\mathbb{R}^{\times}$ is isomorphic to $\mathbb{R}$ as a Lie group, but it cannot appear by itself here, because \eqref{eq:menu} takes the full stabiliser rather than an identity component.

The additive algebraic group has no nontrivial algebraic characters. If $\chi$ is such a character, then $\chi$ is a polynomial satisfying
\[
\chi(s+t)=\chi(s)\chi(t).
\]
Setting $s=t$ gives $\chi(2t)=\chi(t)^2$. If $\chi$ has positive degree $d$, the left side has degree $d$ and the right side has degree $2d$, a contradiction. Therefore $d=0$ and $\chi$ is constant; since $\chi(0)=1$, $\chi\equiv1$. The smooth characters $t\mapsto e^{iqt}$ therefore do not arise from tensorial constructions when the recovered group is additive $\mathbb{R}$: they are smooth but not algebraic. This completes the proof of Proposition~\ref{prop:alg}. Any preimage would recover a full tensor stabiliser isomorphic to $\mathbb{R}$ as a Lie group, hence the additive algebraic group. Its only algebraic character is trivial, while the charged matter representation is nontrivial.

For irrationally related charges, there is a second obstruction to any compact replacement. The image
\[
\{\mathrm{diag}(e^{it},e^{i\alpha t}):t\in\mathbb{R}\}
\]
is dense in $U(1)^2$ when $\alpha\notin\mathbb{Q}$. A polynomial that vanishes on this one-parameter subgroup vanishes on the whole two-torus, so the Zariski closure is $U(1)^2$. Thus the smallest algebraic compact group that carries both charges is two-dimensional. A one-dimensional $U(1)$ replacement exists only for commensurate charges.

\section{An off-menu complexified presentation}\label{app:complex}

Proposition~\ref{prop:fullness} concerns the real oriented object
\[
X_{\mathbb{R}}=(\mathbb{R}^{2m},g,\mathrm{vol}),
\]
whose recovered group is $SO(2m)$. The PFB-side automorphism
\[
g\mapsto RgR^{-1},\qquad R\in O(2m),\quad \det R=-1,
\]
has no preimage in $\mathrm{Hom}_{\mathbf{VB}_{\mathrm{gen}}}(X_{\mathbb{R}},X_{\mathbb{R}})$. An off-menu complex presentation can implement the same improper automorphism. This does not restore fullness; non-fullness then appears in a different Hom-set.

Begin with the original PFB morphism. For the trivial flat witness,
\[
\mathcal{F}(X_{\mathbb{R}})=
(P=M\times SO(2m),\;G=SO(2m),\;\omega_0,\;W=\mathbb{R}^{2m},\;\rho=\mathrm{def}).
\]
With $\eta(g)=RgR^{-1}$, the target category contains
\[
(\Phi,\eta,u):\mathcal{F}(X_{\mathbb{R}})\to\mathcal{F}(X_{\mathbb{R}}),
\]
where
\[
\Phi(x,g)=(x,RgR^{-1}),\qquad u=R:\mathbb{R}^{2m}\to\mathbb{R}^{2m}.
\]
The matter intertwining condition is
\[
u\rho(g)=\rho(\eta(g))u,
\]
which here reads $Rg=(RgR^{-1})R$. The VB side lacks this automorphism because every automorphism of $X_{\mathbb{R}}$ preserves both $g$ and $\mathrm{vol}$, hence lies in $SO(2m)$ and induces only an inner automorphism. The reflection $R$ preserves $g$ but sends $\mathrm{vol}$ to $-\mathrm{vol}$.

Now complexify the presentation. Let
\[
V_{\mathbb{C}}=\mathbb{C}^{2m}
\]
with the real basis inherited from $\mathbb{R}^{2m}$. Let $b$ be the complex-bilinear symmetric form extending $g$, let $h$ be the standard Hermitian form, and let $\epsilon$ be the complexification of $\mathrm{vol}$. Define
\[
X_{\mathbb{C}}=(\mathbb{C}^{2m},b,h,i\epsilon).
\]
Its complex-linear stabiliser is still $SO(2m)$. Preserving $b$ gives $A^{-1}=A^T$, while preserving $h$ gives $A^{-1}=\bar A^T$. Hence $A=\bar A$, so $A$ is a real orthogonal matrix. Preserving $i\epsilon$ then forces $\det A=1$.

Let $c$ be componentwise complex conjugation in the real basis and set
\[
u_{\mathbb{C}}=c\circ R.
\]
This map is conjugate-linear. It is compatible with these tensors in the semilinear sense:
\[
b(u_{\mathbb{C}}x,u_{\mathbb{C}}y)=\overline{b(x,y)},\qquad
h(u_{\mathbb{C}}x,u_{\mathbb{C}}y)=\overline{h(x,y)},
\]
and
\[
(i\epsilon)(u_{\mathbb{C}}x_1,\ldots,u_{\mathbb{C}}x_{2m})=\overline{(i\epsilon)(x_1,\ldots,x_{2m})}.
\]
The last equality uses both sign changes: complex conjugation sends $i$ to $-i$, and the reflection sends $\epsilon$ to $-\epsilon$. Their composite preserves the rescaled volume form in the semilinear sense. Since the matrices of $SO(2m)$ are real in this presentation, $c$ commutes with their action, and $u_{\mathbb{C}}$ induces
\[
g\mapsto RgR^{-1}.
\]
Thus the improper automorphism missing at $X_{\mathbb{R}}$ can be implemented at $X_{\mathbb{C}}$, provided conjugate-linear fibre maps are admitted as morphisms of complex structured fibres.

This does not produce a morphism $X_{\mathbb{R}}\to X_{\mathbb{R}}$. The real object's endomorphism Hom-set is unchanged. The complex presentation instead gives a second source object whose PFB image is isomorphic to that of $X_{\mathbb{R}}$. Take both objects in the vacuum case, with empty matter lists (Section~\ref{sec:setup} allows this). Their stabilisers consist of the same matrices: $\mathsf{Aut}(\mathbb{R}^{2m},g,\mathrm{vol})$ and $\mathsf{Aut}(\mathbb{C}^{2m},b,h,i\epsilon)$ are the group of real orthogonal matrices of determinant one, viewed inside $GL(2m,\mathbb{R})$ and $GL(2m,\mathbb{C})$ respectively. Let $\eta_{\mathrm{can}}$ be the canonical isomorphism sending each matrix to itself, and define the corresponding map of trivial bundles by
\[
\Phi(x,g)=(x,\eta_{\mathrm{can}}(g)).
\]
This map covers the identity on $M$, preserves the flat connection, and is $\eta_{\mathrm{can}}$-equivariant. Hence $\mathcal{F}(X_{\mathbb{R}})$ and $\mathcal{F}(X_{\mathbb{C}})$ are isomorphic. The PFB side does not record whether $SO(2m)$ arose from the real oriented Euclidean fibre or from the complex tensor data $(b,h,i\epsilon)$.

With defining matter included, the isomorphism between the two presentations requires an additional choice in the tensorial grammar. If $\mathcal{F}(X_{\mathbb{C}})$ outputs the full complex defining fibre $\mathbb{C}^{2m}$, it is not automatically isomorphic in $\mathbf{PFB}_{\mathrm{mat}}$ to the real defining fibre $\mathbb{R}^{2m}$. The isomorphism exists if the tensorial grammar for $X_{\mathbb{C}}$ cuts out the real fixed locus of the structure-built conjugation
\[
\sigma=h^{-1}\circ b.
\]
Then $\mathrm{Fix}(\sigma)\cong\mathbb{R}^{2m}$ carries the real defining representation, and the PFB-side isomorphism has
\[
\Phi=\mathrm{id}_{M\times SO(2m)},\qquad \eta=\eta_{\mathrm{can}},\qquad
u:\mathbb{R}^{2m}\to\mathrm{Fix}(\sigma),\quad u(x)=x
\]
in the chosen real basis.

The distinctions are summarised in Table~\ref{tab:presentation-morphisms}.

\begin{center}
\renewcommand{\arraystretch}{1.35}
\begin{tabular}{>{\raggedright\arraybackslash}p{0.31\textwidth}|>{\raggedright\arraybackslash}p{0.43\textwidth}|>{\raggedright\arraybackslash}p{0.16\textwidth}}
\textbf{PFB isomorphism} & \textbf{Formula} & \textbf{VB preimage?}\\
\hline
$\mathcal{F}(X_{\mathbb{R}})\to\mathcal{F}(X_{\mathbb{R}})$ & $\Phi(x,g)=(x,RgR^{-1})$, $\eta=\operatorname{Ad}_R$, $u=R$ & No.\\
$\mathcal{F}(X_{\mathbb{R}})\to\mathcal{F}(X_{\mathbb{C}})$, vacuum & $\Phi(x,g)=(x,\eta_{\mathrm{can}}(g))$, $\eta=\eta_{\mathrm{can}}$ & No.\\
$\mathcal{F}(X_{\mathbb{R}})\to\mathcal{F}(X_{\mathbb{C}})$, real-form matter & $\Phi(x,g)=(x,\eta_{\mathrm{can}}(g))$, $\eta=\eta_{\mathrm{can}}$, $u:\mathbb{R}^{2m}\cong\mathrm{Fix}(\sigma)$ & No.\\
$\mathcal{F}(X_{\mathbb{C}})\to\mathcal{F}(X_{\mathbb{C}})$ & $\Phi(x,g)=(x,RgR^{-1})$, $\eta=\operatorname{Ad}_R$, $u=c\circ R$ & Yes, if semilinear VB morphisms are allowed.\\
\end{tabular}
\captionof{table}{Four PFB-side isomorphisms involving the real and complexified presentations. The map $c\circ R$ implements the improper automorphism at $X_{\mathbb{C}}$; it is not the isomorphism $\mathcal{F}(X_{\mathbb{R}})\to\mathcal{F}(X_{\mathbb{C}})$ between the two presentations.}\label{tab:presentation-morphisms}
\end{center}

The PFB isomorphism between the real and complexified presentations is therefore not $c\circ R$. It is $(\Phi,\eta_{\mathrm{can}})$ as defined above, together with a matter intertwiner only if the complexified presentation outputs the structure-built real fixed locus. The map $c\circ R$ instead realises the improper automorphism at $X_{\mathbb{C}}$.

The complexified presentation therefore does not restore fullness of the original functor. Once $X_{\mathbb{C}}$ is admitted, the improper automorphism has a source preimage at that object, but the source then contains two non-isomorphic VB objects whose PFB images can be isomorphic.\footnote{A simpler presentation-multiplicity case occurs already at $m=1$, on the classical menu: the real oriented plane $(\mathbb{R}^2,g,\mathrm{vol})$ and the Hermitian line $(\mathbb{C},\langle\cdot,\cdot\rangle)$ present isomorphic vacuum PFB theories, since $SO(2)\cong U(1)$, but they are not isomorphic as VB-generating objects. This case is not used in Proposition~\ref{prop:fullness}, which assumes $m\geq2$ because complex conjugation at the Hermitian-line presentation implements the inversion automorphism.
} Non-fullness then appears in the comparison
\[
\mathrm{Hom}_{\mathbf{VB}_{\mathrm{gen}}}(X_{\mathbb{R}},X_{\mathbb{C}})
\longrightarrow
\mathrm{Hom}_{\mathbf{PFB}_{\mathrm{mat}}}\bigl(\mathcal{F}(X_{\mathbb{R}}),\mathcal{F}(X_{\mathbb{C}})\bigr).
\]
In the vacuum case the source Hom-set is empty: $\dim_{\mathbb{R}}\mathbb{R}^{2m}=2m$ while $\dim_{\mathbb{R}}\mathbb{C}^{2m}=4m$, so no invertible fibrewise map between the two objects exists, of any type. The target Hom-set, by contrast, contains the canonical isomorphism. A full functor cannot map an empty Hom-set onto a non-empty one. Equivalently, a full and faithful functor is injective on isomorphism classes. Given an isomorphism $\mathcal{F}(X)\to\mathcal{F}(X')$, fullness supplies source morphisms mapping to it and to its inverse; faithfulness then shows that their composites are the identity morphisms. Hence two non-isomorphic presentations with isomorphic PFB images already preclude equivalence. With defining matter included, the same conclusion holds under the real-fixed-locus convention just described. Complexification changes the Hom-set in which non-fullness appears; it does not restore fullness.

\paragraph{The order-two bound.} At any fixed presentation, structure-compatible semilinear maps can realise at most one nontrivial outer class. Let $V$ be a fibre over $k\in\{\mathbb{R},\mathbb{C}\}$ carrying tensor data $T$, and let $u:V\to V$ be invertible, $\sigma$-semilinear---$u(\lambda x)=\sigma(\lambda)u(x)$ for a continuous field automorphism $\sigma$ of $k$---and structure-compatible in the semilinear sense used above. Each such $u$ induces the automorphism $\eta_u(g)=ugu^{-1}$ of $G=\mathsf{Aut}(V,T)$. If $\sigma=\mathrm{id}$, then $u\in G$ by \eqref{eq:menu}, and $\eta_u$ is inner. If $k=\mathbb{C}$ and $\sigma\neq\mathrm{id}$, then $\sigma$ is complex conjugation. Hence $\sigma^{2}=\mathrm{id}$, the composite $u^{2}$ is linear and structure-compatible, and $u^{2}\in G$. In the outer automorphism group $\mathrm{Out}(G)=\mathrm{Aut}(G)/\mathrm{Inn}(G)$,
\[
[\eta_u]^{2}=[\eta_{u^{2}}]=[\mathrm{id}].
\]
Moreover, if $u$ and $v$ are two conjugate-linear structure-compatible maps at the same presentation, then $v^{-1}u$ is linear and structure-compatible, hence belongs to $G$. They therefore induce the same outer class. For $k=\mathbb{R}$ there is no nontrivial $\sigma$ (Section~\ref{sec:functor}). Thus, for the scalar fields admitted in this paper, the outer classes realised by structure-compatible semilinear maps at a fixed presentation form a group of order at most two. No semilinear enlargement can realise an outer automorphism of order three.

The triality automorphism of $\mathrm{Spin}(8)$ has order three and cyclically permutes its three eight-dimensional irreducible representations. Once the menu is enlarged to present $\mathrm{Spin}(8)$, the trivial flat vacuum object has this target automorphism, but no semilinear source morphism can induce it.

\end{document}